\documentclass[11pt]{article}
\usepackage{algorithm}
\usepackage{algorithmic}
\usepackage{graphicx}
\usepackage{hyperref}
\usepackage{amsmath}
\usepackage{amssymb}
\usepackage{textcomp}
\usepackage{amsthm}
\usepackage{algorithm,algorithmic}
\usepackage{latexsym,amssymb,amsfonts,amsmath}
\usepackage{graphicx}
\usepackage{hyperref}
\usepackage[usenames,dvipsnames]{pstricks}
\usepackage{epsfig}
\usepackage{pst-grad}
\usepackage{pst-plot}
\usepackage{verbatim}
\hypersetup{
	colorlinks,
	linktoc=all,
	pdfstartview={FitH}
}
\newtheorem{theorem}{Theorem}
\newtheorem{corollary}{Corollary}
\newtheorem{lemma}{Lemma}
\newtheorem{claim}{Claim}

\def \bclm {\begin{claim}} 
\def \eclm {\end{claim}} 
\def \bitem {\begin{itemize}} 
\def \eitem {\end{itemize}} 
\def \bprf {\begin{proof}} 
\def \eprf {\end{proof}} 
\newcommand{\problemdefn}[3]{
    \begin{center}
    \noindent
    \framebox{
        \begin{minipage}{4.6in}
            \textbf{\sc #1} \\
            \emph{Input}: #2 \\
            \emph{Output}: #3
        \end{minipage}
    }
    \end{center}
} 

\def \call {{\cal L}} 
\def \cals {{\cal S}}

\begin{document}
\title{Approximating MIS over equilateral $B_1$-VPG 
graphs\footnote{Journal version of a part of the {\em preliminary} version presented at COCOA-2015 \cite{lahiri}.} }
\author{Abhiruk Lahiri\thanks{Indian Institute of Science, Bangalore, India; 
abhiruk.lahiri@csa.iisc.ernet.in}
 ~~~~ 
Joydeep Mukherjee\thanks{The Institute of Mathematical Sciences, HBNI, Chennai, India;  
joydeepm@imsc.res.in} 
 ~~~ 
C.R. Subramanian\thanks{
The Institute of Mathematical Sciences, HBNI, Chennai 600113, India;  
crs@imsc.res.in }} 

\maketitle

\begin{abstract}

We present an approximation algorithm for the maximum independent set 
(MIS) problem over the class of equilateral $B_1$-VPG graphs. These are 
intersection graphs of $L$-shaped planar objects 
with both arms of each object being equal. 
We obtain a $36(\log 2d)$-approximate algorithm running in 
$O(n(\log n)^2)$ time for this problem, 
where $d$ is the ratio $d_{max}/d_{min}$ and $d_{max}$ and $d_{min}$ denote 
respectively the maximum and minimum length of any arm in the input 
equilateral $L$-representation of the graph. 
In particular, we obtain 
$O(1)$-factor approximation of MIS for $B_1$-VPG -graphs for which the 
ratio $d$ is bounded by a constant. 
In fact,  algorithm can be generalized to an $O(n(\log n)^2)$ time and 
a $36(\log 2d_x)(\log 2d_y)$-approximate MIS algorithm over arbitrary $B_1$-VPG graphs.  
Here, $d_x$ and $d_y$ denote respectively the analogues of $d$ when restricted to 
only horizontal and vertical arms of members of the input. 
This is an improvement over the previously best $n^\epsilon$-approximate algorithm 
\cite{FoxP} (for some fixed $\epsilon>0$), unless the ratio $d$ is exponentially 
large in $n$. In particular, $O(1)$-approximation of MIS is achieved for graphs 
with $\max\{d_x,d_y\}=O(1)$. 
\end{abstract}


 \section{Introduction} 
 The problem of computing an approximation to a maximum independent set ({\em MIS}) of 
 an arbitrary graph is notoriously hard. It is known \cite{hastad1999} that, 
 for every fixed $\epsilon > 0$, MIS cannot 
 be efficiently approximated (unless $NP=ZPP$) within a multiplicative factor of $|V|^{1-\epsilon}$ for an 
 arbitrary $G=(V,E)$. 
Naturally, there have been algorithmic studies of this problem on special classes of graphs. 
One such graph class is denoted by $B_1$-VPG. 

\vspace{2mm} 

{\em Vertex intersection graphs of Paths on Grid} {thm:aprx-mis-b1-vpg-gen} 
(or, in short, VPG graphs) were first introduced 
by Asinowski et. al. \cite{Asin}. 
For a member of this class of graphs, its vertices represent 
paths joining grid-points on a rectangular planar grid and two such 
vertices are adjacent if and only if the corresponding paths intersect 
in a grid-point. 
In particular, $B_1$-VPG graphs denotes the class of 
intersection graphs of paths on a grid where each path has one of the 
following shapes: $\llcorner, \ulcorner, \urcorner$ and $\lrcorner$, commonly 
referred to as an $l$.
By an \emph{arm} of an $l$, we mean either a horizontal or a vertical line segment
associated with $l$. An $l$ is said to be {\em equilateral} if its vertical and horizontal 
arms are of equal length $le(l)$. 
In this paper, we focus on equilateral $B_1$-VPG graphs formed 
by equilateral $l$'s. For a set $\call$ of equilateral $l$'s, we denote by $d_{max}(\call)$ 
and $d_{min}(\call)$ the maximum and minimum values of $le(l)$ over $l \in \call$. 

\vspace{2mm} 

VPG graphs are a special type of \emph{string graphs}, which are 
intersection graphs of simple curves in the plane \cite{Asin}. 
The best known approximation 
algorithm for MIS on string graphs is only known to 
have a guarantee of $n^{\epsilon}$, for some $\epsilon > 0$ \cite{FoxP}. 
There are few subclasses of string graphs like 
outerstring graphs, planar graphs, $B_1$-EPG, etc., classes which admit 
efficient computation of either a MIS or a constant factor approximation 
of a MIS \cite{bakerJACM1994,epst,keil17}. 
Recently, in \cite{vpgipl17}, 
we presented a $O(n(\log n)^3)$-time, $4(\log n)^2$-approximation algorithm 
for MIS over $B_1$-VPG graphs on $n$ vertices.

  ~~ 

\noindent{\bf Our Results : } 
In this paper, we present  
new approximation algorithms for the class of equilateral 
$B_1$-VPG graphs. 
The decision version of this problem is 
NP-complete even if restricted to instances where the arms of 
of all $l$'s are of equal length (see Theorem $5$ in \cite{lahiri}). 
Throughout the paper, 
we assume that the input is in the form of a set $\call$ of equilateral $l$'s. 
Precisely, we obtain the following results.
\begin{theorem} \label{thm:aprx-mis-b1-vpg} 
 There exists an $O(n(\log n)^2)$-time and $36(\log 2d)$-approximation 
 algorithm for 
MIS restricted to equilateral $B_1$-VPG graphs, where 
$d=d(\call)=d_{max}(\call)/d_{min}(\call)$. 
\end{theorem} 
In particular, for all 
equilateral $\call$ with $d(\call)=O(1)$, 
the algorithm of Theorem \ref{thm:aprx-mis-b1-vpg} 
yields an $O(1)$-approximation of MIS. 
Also, when this is result combined with the approximation algorithm 
of \cite{vpgipl17}, we obtain a slightly slower $O(n(\log n)^3)$-time and 
$\min\{4(\log n)^2, 36(\log 2d)\}$-approximation algorithm 
for MIS over equilateral $B_1$-VPG graphs. 

  ~~~~ 

When all members have a uniformly common arm length, 
the following corollary can be inferred by slightly 
modifying the proof arguments 
of Theorem \ref{thm:aprx-mis-b1-vpg}. We provide a brief sketch of 
of this modification in Section \ref{sec:proofsofcorollaries}.  
\begin{corollary} \label{cor:mis-aprx-unit-b1-vpg} 
 There exists an $O(n(\log n)^2)$-time and $16$-approximation 
 algorithm for computing a MIS over equilateral $B_1$-VPG graphs 
 formed by sets of $l$'s of uniformly common arm lengths. 
\end{corollary} 
It is easy to see from the description of the algorithm and its 
analysis that the algorithm can be suitably generalized to 
obtain an approximate MIS algorithm over {\em arbitrary} $B_1$-VPG 
graphs. Precisely, we obtain the following theorem. A brief sketch of 
its proof is provided in Section \ref{sec:03}. 
\begin{theorem} \label{thm:aprx-mis-b1-vpg-gen} 
 There exists an $O(n(\log n)^2)$-time and $36(\log 2d_x)(\log 2d_y)$ 
 approximation algorithm for 
MIS restricted to $B_1$-VPG graphs, where 
$d_x=d_{max}^x(\call)/d_{min}^x(\call)$, 
$d_y=d_{max}^y(\call)/d_{min}^y(\call)$. $d_{max}^x(\call)$ is the maximum 
length of the horizontal arm of any member of $\call$. $d_{min}^x$, $d_{max}^y$ and 
$d_{min}^y$ are similarly defined. 
\end{theorem}
When combined with the approximation algorithm of 
\cite{vpgipl17}, this yields an $O(n(\log n)^3)$-time and 
$\min\{4(\log n)^2, 36(\log 2d_x)(\log 2d_y)\}$ approximation algorithm 
for MIS over arbitrary $B_1$-VPG graphs. 
In particular, for $B_1$-VPG graphs having $d_x,d_y=O(1)$, we obtain 
an efficient, $O(1)$-approximation algorithm for MIS. 
To the best of our knowledge, no such $O(1)$-factor approximation of MIS is known for any class of 
$B_1$-VPG graphs. 
We also infer the following corollary 
by slightly modifying the proof arguments of Theorem \ref{thm:aprx-mis-b1-vpg-gen}. 
A sketch of this modification is provided in Section \ref{sec:proofsofcorollaries}
\begin{corollary} \label{cor:mis-aprx-unit-b1-vpg-gen} 
 There exists an $O(n(\log n)^2)$-time and $16$-approximation 
 algorithm for computing a MIS over {\em arbitrary} $B_1$-VPG graphs 
 formed by sets of $l$'s having a uniformly common arm length for {\em each} of the 
 horizontal and vertical arms of all members $\call$. The vertical and horizontal arm lengths 
 may however differ for any $l$ in the input $\call$. 
\end{corollary} 

We introduce some conventions and notations in Section \ref{sec:02}. 
In Section \ref{sec:03}, we present the MIS approximation 
algorithm and its analysis for equilateral $B_1$-VPG graphs. 
This constitutes the proof of Theorem \ref{thm:aprx-mis-b1-vpg}. In 
Section \ref{sec:04}, we present a sketch of the generalization of 
the approach of Section \ref{sec:03} to 
arbitrary $B_1$-VPG graphs and its analysis. This constitutes the 
proof of Theorem  \ref{thm:aprx-mis-b1-vpg-gen}. 
In Section \ref{sec:proofsofcorollaries}, we provide a sketch of 
the proof arguments of Corollaries \ref{cor:mis-aprx-unit-b1-vpg} and 
\ref{cor:mis-aprx-unit-b1-vpg-gen}. 
In Section \ref{sec:05}, we conclude with some remarks.

\section{Preliminaries}
\label{sec:02}
We work with geometric shapes $\llcorner, \ulcorner, \urcorner$ and $\lrcorner$. 
For ease of further discussion, we refer to them as follows. $L_1$ refers to the 
shape $\llcorner$, $L_2$ refers to $\ulcorner$,$L_3$ refers to $\urcorner$ and $L_4$ to 
$\lrcorner$.
Henceforth, we use $l$ to denote a geometric object with one of the 
four shapes $L_1,L_2,L_3$ and $L_4$. 

The \emph{corner} of an $l$ is defined to be the point where the two arms meet 
and is denoted by $c_l$, the tip of the 
horizontal arm is denoted by $h_l$ and that of the vertical arm is 
denoted by $v_l$. For an object $l$, we use $(x_c,y_c,x_h,y_v)$ to 
denote respectively the $x$- and $y$- coordinates of $c_l$, 
the $x$- coordinate of $h_l$ and the $y$- coordinate of $v_l$. 
This 4-tuple completely and uniquely describes $l$.
The set of points constituting $l$ is denoted by $P_l$ 
and is given by (when $l$ is of shape $L_1$) 
$$P_{l} = \{(x,y_c) \; : \; x_c \leq x \leq x_h\} \cup 
\{(x_c,y) \; : \; y_c \leq y \leq y_v\}.$$ 
We say that two distinct objects $l_1$ and $l_2$ intersect if 
$P_{l_1}\cap P_{l_2} \neq \emptyset$.  Given a set $\call$ of $l$'s, the 
intersection graph $G$ formed by $\call$ is defined to be 
$G=(\call, E)$ where $E$ consists of all those unordered pairs 
$(l_1,l_2)$ such that $l_1$ and $l_2$ intersect. 
A set of $l$'s such that no two of them form an 
intersecting pair is said to be an \emph{independent} set.

For each $1 \leq i \leq 4$, we refer to an intersection graph formed by 
objects  each of shape $L_i$ as a $L_i$-graph. 
For ease of description, we refer to a $L_1$-graph as a $L$-graph. 
By symmetry (based on rotations), one 
can adapt any efficeint and exact/approximate MIS algorithm for $L_1$-graphs to 
a similar algorithm (with same time and approximation guarantee) for $L_i$-graphs,  
for every $i$.  This enables us to focus only on $L_1$-graphs (at the cost of 
increasing the approximation guarantee by a multiplicative factor of 4) as is established 
by the following Claim \ref{clm:Lsuffices}.  
The claim is essentially a formal statement (in the context of equilateral $B_1$-VPG graphs) 
of an approach that has been employed in \cite{epst} for $B_1$-EPG graphs. 
Let $\alpha(G)$ denote the size of a MIS of graph $G$ and 
 let $A(G)$ denote the size of an independent set of $G$ returned by an 
algorithm $A$.
\begin{claim} \label{clm:Lsuffices} 
If there is an efficient algorithm $A$ which approximates MIS over equilateral $L$-graphs 
within a multiplicative factor of $c(d(\call))$ $($for some increasing function $c(d))$, 
then there is an efficient algorithm $B$ which approximates MIS over 
equilateral $B_1$-VPG graphs within a multiplicative factor 
of $4c(d(\call))$.  Throughout, $\call$ denotes the input to the corresponding algorithm. 
\end{claim} 
Hence, from now on, we focus only on the subclass of equilateral $L$-graphs.   

All logarithms used below are with respect to base 2. We denote a set $\{1,2,\ldots,n\}$ by $[n]$.
A permutation over $[n]$ is any bijection 
$\pi :[n] \rightarrow [n]$. An inversion of $\pi$ is any 
unordered pair $(i,j)$ satisfying $(i-j)(\pi^{-1}(i)-\pi^{-1}(j))<0$. 
The graph $G_{\pi}$ associated with $\pi$ is defined as $G_{\pi}=([n],E)$ 
where $E$ is set of all inversions of $\pi$. A permutation graph on 
$n$ vertices is any graph which is isomorphic to $G_{\pi}$ for some permutation 
$\pi$. 

\section{MIS-approximation over equilateral VPG-graphs}
\label{sec:03}

In this section, we prove Theorem \ref{thm:aprx-mis-b1-vpg} by designing an approximation 
algorithm for the following problem and then applying Claim \ref{clm:Lsuffices}. 

\def \logtd {{(\lfloor \log 2d \rfloor)}} 

\problemdefn{Maximum Independent Set over equilateral $L$-graphs}
{A set $\mathcal{L}$ of $L_1$-shaped equilateral $l$'s }{a set $I \subset \mathcal{L}$  such that $I$ is independent and $|I|$ is 
maximized.}

\noindent 
Before we present an approximation algorithm for 
this problem, 
we present the following claim (proved and employed in \cite{vpgipl17}). 
\begin{claim} \label{clm:assum} 
Without loss of generality, we can assume that : \\
 $l_1.x_c \neq l_2.x_c$ and $l_1.y_c \neq l_2.y_c$ for any pair of 
distinct $l_1,l_2 \in \cals$, where $\call$ is the input set of $L_1$-shaped 
$l$'s.  
\end{claim} 
We will also be making use of the following lemma (Lemma 1 of \cite{vpgipl17}) in 
the design of new approximation algorithms. 
\begin{lemma} (\cite{vpgipl17}) \label{lem:vhpermgraph}  
 Suppose $S'$ is a set of $l$'s, each being of type $L_1$. Suppose 
there exist a horizontal line $y=b$ and a vertical line $x=a$ such 
that each $l \in S'$ intersects both $y=b$ and $x=a$. 
Then, the intersection graph of members of $S'$ is a permutation graph.
\end{lemma} 
We will also be making use of the following fact from \cite{kim1990}. 
\begin{theorem} \cite{kim1990} \label{thm:permgrMIS}
Given an arbitrary permutation graph $G$ (in the form of two 
permutations defining $G$), a MIS of $G$ can be computed in 
$O(n(\log n))$ time where $n=|V(G)|$. 
\end{theorem} 
We will also be making use of the following claim. 
\bclm \label{claim : unit-leng} 
Without loss of generality, assume that input $\call$ satisfies 
$d_{min}=2$. 
\eclm 
\bprf (sketch:) Rescale the the coordinates of $x$-axis and $y$-axis 
by stretching both of them by a multiplicative factor of 
$2/d_{min}$.     
\eprf  
The algorithm begins by dividing the input set $\call$ into disjoint sets 
$S_1,S_2,$ $\ldots,S_{\lfloor \log 2d \rfloor}$ where 
$S_i=\{l\in S \; | \; 2^i \leq le(l)< 2^{i+1}\}$, 
$\forall i \in [\lfloor \log 2d \rfloor]$. This split is 
to exploit the fact that $d(S_i) \leq 2$, for any $i$. 
We further partition each $S_i$ into nine subsets as follows. 

For the $i^{th}$ set, we do the following. 
We place a sufficiently large but finite grid structure on the 
plane covering all members of $S_i$. 
The grid is chosen in such a way so that grid-length in each of the 
$x$ and $y$ directions is $2^i$. What we get is a rectangular array of
square boxes of side length $2^i$ each. We number the rows of boxes 
from bottom and the columns of boxes from left, with numbers $0,1, \ldots$. 
  
We denote a \emph{box} by $(r',c')$ if it is in the intersection of 
${r'}^{th}$ row and ${c'}^{th}$ column. We say $l$ {\em lies inside} a box 
if its corner $c_l$ lies either in the interior or on the boundary of the box, 
except that it should not lie either on its top horizontal boundary  
or on its right vertical boundary. 
If $l$ lies inside a box $(r',c')$ we denote it by $l \in (r',c')$. 

For every  $k_r,k_c \in \{0,1,2\}$,  
define $$S_{i,k_r,k_c}=\{l\in (r',c') \; | \; 
r' \cong k_r  \; {\rm  mod } \; 3, \; \; c'\cong k_c \; {\rm mod} \;  3\}.$$ 
For a pictorial representation of the partition of $S_i$ 
into 9 subsets, see the figure below. 
\begin{figure}[h]
\centering
 \includegraphics[scale=0.75]{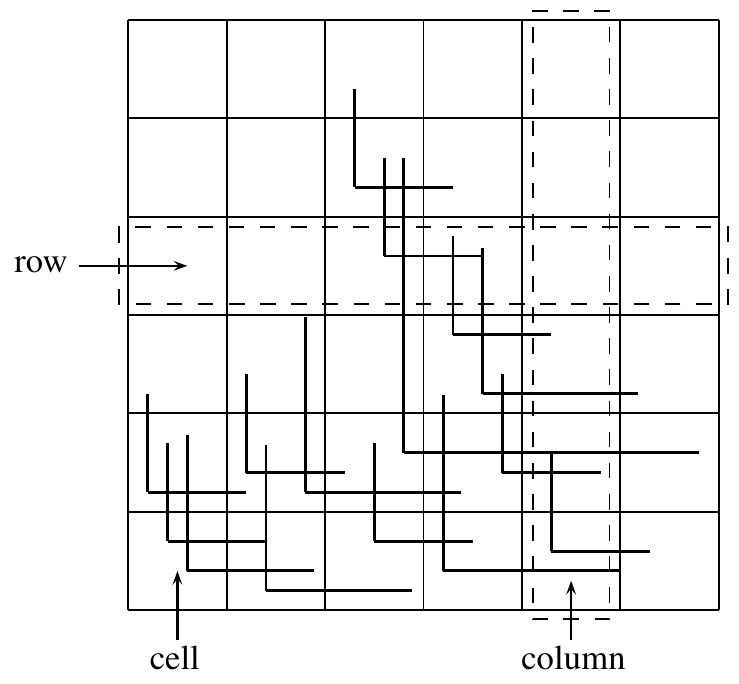}
\caption{The grid is for $l$'s of type $1$ whose lengths lie within the range $[2^i, 2^{i+1})$.}
\label{fig:01}
\end{figure} 

Thus, we partition in the input into induced 
subgraphs $G_1, \ldots ,G_D$ where $D=9 \logtd$.  
Each $G_j$ is a subgraph induced by $S_{i,k_r,k_c}$ for some $i \in 
[\lfloor \log 2d \rfloor]$ and $k_r,k_c \in \{0,1,2\}$. 
In Lemma \ref{eqlem2}, we establish that a MIS can be computed 
in efficiently for the intersection graph  
$G(S_{1,0,0})$ induced by $S_{1,0,0}$ and hence, 
by symmetry, a MIS can be computed for 
each of the $9(\lfloor \log 2d \rfloor)$ 
induced subgraphs.  
More precisely, for each of the $9(\lfloor \log 2d \rfloor)$  induced subgraphs, 
a MIS can be computed in $O(p(\log p)^2)$ time, where $p$ represents 
the number of vertices in the respective induced subgraph, leading to 
an $O(n(\log n)^2)$ overall running time. 
We choose the largest of these $9(\lfloor \log 2d \rfloor)$  independent sets 
and return it as the output. Since $\alpha(G) \leq \sum_{1 \leq i \leq D} \alpha(G_i)$, 
we deduce that the algorithm just outlined above returns an independent set of 
size at least $\alpha(G)/9 \logtd$ for an arbitrary equilateral $L$-graph $G$. 
Now, combining this observation with Claim \ref{clm:Lsuffices}, 
we deduce Theorem \ref{thm:aprx-mis-b1-vpg}. 

It now remains only to prove that a MIS can be efficiently computed for 
$G(S_{1,0,0})$,  
as stated in Lemma \ref{eqlem2} below. 
\begin{lemma}\label{eqlem2}
 A MIS can be computed in $O(p(\log p)^2)$ time for $G(S_{1,0,0})$. 
 Here, $p=|S_{1,0,0}|$.
\end{lemma}
\begin{proof} 
 Recall that $S_{1,0,0}=\{l \in S_1 : l \in (r',c'), r',c' \cong 0 \; {\rm  mod } \; 3\}$. 
From the definitions of boxes given above, the following can be seen immediately : 
if $r_1,r_2,c_1,c_2 \cong 0 \; {\rm mod} \; 3$ and $(r_1,c_1) \not= (r_2,c_2)$, 
then for any $l_1 \in (r_1,c_1)$ and any 
$l_2 \in (r_2,c_2)$, $l_1$ and $l_2$ are independent. Hence, computing a 
MIS for $H=G(S_{1,0,0})$ reduces to computing, for each $(r,c)$ such that 
$r,c \cong 0 \; {\rm mod} \; 3$, a MIS for the subgraph 
of $H$ induced by those $l \in (r,c)$. Since, for any such $(r,c)$, 
each $l \in (r,c)$ intersects the vertical line forming the left-border 
of the $(c+1)$-th column of boxes as well as the horizontal line 
forming the bottom-border of the $(r+1)$-th row of boxes, we 
can deduce, by applying Lemma \ref{lem:vhpermgraph} of \cite{vpgipl17} 
stated before, that the subgraph of $H$ induced by $l \in (r,c)$ forms a 
permutation graph. 

Hence, by applying Theorem \ref{thm:permgrMIS}, 
it follows that a MIS can be computed in $O(q(\log q)^2)$ 
time where $q=|\{l \in S_1 \; : \;  l \in (r,c)\}|$, provided 
the subgraph induced by $l \in (r,c)$ is specified in the form of two 
permutations defining it. It follows from the proof of Lemma \ref{lem:vhpermgraph}
that the two permutations specifying the input are 
the two increasing orders formed by the $l.x_c$ and $l.y_c$ values of its members. 
These two orders can be computed in $O(q(\log q))$ time. 
 \end{proof}

 \section{Generalization to arbitrary $B_1$-VPG graphs} 
 \label{sec:04}. 

 \noindent{\bf Proof of Theorem \ref{thm:aprx-mis-b1-vpg-gen} :} 
 
    ~~~~  
 
 The broad approach is the same as for equilateral graphs except that 
 we partition the members based on their lengths in each of the vertical and 
 horizontal directions independently. This independent partitioning was not 
 needed for the equilateral case since lengths are the same for both arms of 
 any $\ell$. 
 
 As earlier, we assume (without loss of generality) that $d^x_{min}=2$ and 
 $d^y_{min}=2$. 
 We divide horizontal lengths lying in $[d^x_{min}, d^x_{max}]$ into groups 
 $[2^{i}, 2^{i+1})$ for $i \in [\lfloor \log 2d_x\rfloor]$ and also  
 divide vertical lengths lying in $[d^y_{min}, d^y_{max}]$ into groups 
 $[2^{i}, 2^{i+1})$ for $j \in [\lfloor \log 2d_y\rfloor]$. 
 Using these two partitions, we divide members of $\call$ into 
 sets $S_{i,j}$ with each $S_{i,j}$ consisting of those members of $\call$ 
 whose horizontal and vertical lengths lie in groups $[2^i, 2^{i+1})$ and 
 $[2^j, 2^{j+1})$ respectively. 

 We further divide each of these $(\lfloor \log 2d_x\rfloor)(\lfloor \log 2d_y\rfloor)$ 
 groups $S_{i,j}$ into 
 9 smaller groups $S_{i,j,k_r, k_c}$ for $k_r,k_c \in \{0,1,2\}$ by imposing a rectangular 
 grid structure with grid-points being separated by lengths $2^i$ and $2^j$ in the 
 horizontal and vertical directions respectively. The remaining details are as before leading 
 to an algorithm running in $O(n(\log n)^2)$ time and producing a  
 $36(\log 2d_x)(\log 2d_y)$ MIS-approximation. 
 
 \section{Proofs of Corollaries \ref{cor:mis-aprx-unit-b1-vpg} and \ref{cor:mis-aprx-unit-b1-vpg-gen}} 
 \label{sec:proofsofcorollaries} 

 \noindent{\bf Proof sketch of Corollary \ref{cor:mis-aprx-unit-b1-vpg} :} 
 
     ~~~ 
 
   The algorithm is essentially the same as the one described in the proof of 
   Theorem \ref{thm:aprx-mis-b1-vpg} except for the following changes. 
   When all equilateral $l$'s in the input $\call$ have uniformly the 
   same arm length $a$, we have only one group instead of $(\lfloor \log 2d \rfloor)$
groups we had in the proof of Theorem \ref{thm:aprx-mis-b1-vpg}. In this case, 
it suffices to impose a finite grid structure whose side lengths (both horizontal 
and vertical) are $a$. After numbering the rows and columns of boxes, we partion 
the input into 4 subsets $S_{k_r,k_c}$ with $k_r,k_c \in \{0,1\}$ where, as before, 
$$S_{k_r,k_c} = \{l\in (r',c') \; | \; 
r' \cong k_r  \; {\rm  mod } \; 2, \; \; c'\cong k_c \; {\rm mod} \;  2\}.$$ 
For each $k_r,k_c \in \{0,1\}$, MIS can be computed exactly for the subgraph 
induced by $S_{k_r,k_c}$. The reason is the same as before : 
 for any $(r_1,c_1) \neq (r_2,c_2)$ satisfying 
$r_1,r_2 \cong k_r  \; {\rm  mod } \; 2,$ $ c_1,c_2 \cong k_c \; {\rm mod} \;  2$
and for any $l_1 \in (r_1,c_1)$ and $l_2 \in (r_2,c_2)$, $l_1$ and $l_2$ are 
independent and hence it reduces to computing, for each such $(r,c)$, exactly an  MIS for 
the subgraph induced by those $l \in (r,c)$ and this can be realized in $O(p(\log p)^2)$ time 
as explained before, where $p=|\{l \in (r,c)\}|$. The largest of the 4 MIS's (one for each 
$S_{k_r,k_c}$) is then returned as the output yielding a 4-approximation of MIS for 
the subgraph induced by $\ell$'s of Type 1. When combined with Claim \ref{clm:Lsuffices}, 
we obtain Corollary \ref{cor:mis-aprx-unit-b1-vpg}. 

   ~~~ 
   
   \noindent{\bf Proof sketch of Corollary \ref{cor:mis-aprx-unit-b1-vpg-gen} :} 
 
     ~~~ 
 
   When all equilateral $l$'s in the input $\call$ have uniformly the 
   same vertical arm length $a$ and horizontal arm length $b$, 
   we have only one group instead of $(\lfloor \log 2d_x\rfloor)(\lfloor \log 2d_y\rfloor)$
groups we had in the proof of Theorem \ref{thm:aprx-mis-b1-vpg-gen}. In this case, 
it suffices to impose a finite grid structure whose vertical and horizontal side lengths
are $a$ and $b$ respectively. After numbering the rows and columns of boxes, we partion 
the input into 4 subsets $S_{k_r,k_c}$ with $k_r,k_c \in \{0,1\}$ where, as before, 
$$S_{k_r,k_c} = \{l\in (r',c') \; | \; 
r' \cong k_r  \; {\rm  mod } \; 2, \; \; c'\cong k_c \; {\rm mod} \;  2\}.$$ 
As explained before, for each $k_r,k_c \in \{0,1\}$, MIS can be computed exactly for the subgraph 
induced by $S_{k_r,k_c}$ leading to a 16-approximation of MIS in polynomial time for 
arbitrary $B_1$-VPG graphs having a uniformly common vertical and horizontal arm lengths. 
This establishes Corollary \ref{cor:mis-aprx-unit-b1-vpg-gen}. 

\section{Conclusions}
\label{sec:05}  

In further works \cite{crsb2vpgsep16}, we have obtained further improvements on MIS approximation 
of $B_1$-VPG graphs and also for improved MIS approximation algorithms $B_2$-VPG graphs.  
It would be interesting to establish some inapproximability results for the MIS problem 
over equilateral $B_1$-VPG graphs. Also the question of obtaining better approximations in 
terms of ratios of lengths would be worth pursuing. 

 ~~~~~ 
 
  \noindent{\bf Acknowledgements :} We thank an anonymous referee (of a related 
  submission) for pointing out that 
a special type of input studied here actually induces a permutation graph, thereby 
leading to improved running time bounds.  

\bibliography{res,oth}
\bibliographystyle{elsarticle-num}

\end{document}